\documentclass[draftclsnofoot, onecolumn, 10pt, one-side]{IEEEtran}
\usepackage{times,amsmath,color,amssymb,graphicx,epsfig,cite,psfrag,enumerate,amsthm,multicol,subfigure}

\theoremstyle{remark}
\newtheorem{Remark}{Remark}
\theoremstyle{plain}
\newtheorem{Theorem}{Theorem}

\theoremstyle{definition}
\newtheorem{Definition}{Definition}

\allowdisplaybreaks

\begin{document}
\title{Achievable Rate Regions for Discrete Memoryless Interference Channel with State Information}
\author{Lili~Zhang, 
        Jinhua~Jiang, 
        and~Shuguang~Cui
\thanks{L. Zhang and S. Cui are with the Department
of Electrical and Computer Engineering, Texas A$\&$M University,
College Station, TX 77843 USA
(e-mail: lily.zhang@tamu.edu; cui@ece.tamu.edu).}
\thanks{J. Jiang is with the Department of Electrical Engineering, Stanford University, Stanford, CA 94305 USA (email: jhjiang@stanford.edu).}
}
\maketitle
\begin{abstract}
In this paper, we study the state-dependent two-user interference channel, where the state information is non-causally known at both transmitters but unknown to either of the receivers. We propose two coding schemes for the discrete memoryless case: simultaneous encoding for the sub-messages in the first one and superposition encoding in the second one, both with rate splitting and Gel'fand-Pinsker coding. The corresponding achievable rate regions are established.
\end{abstract}
\IEEEpeerreviewmaketitle
\renewcommand{\baselinestretch}{1}
\section{Introduction} \label{sec_intro}
The interference channel (IC) models the situation where several independent transmitters communicate with their corresponding receivers over a common channel. Due to the shared medium, each receiver suffers from interferences caused by the transmissions of other transceiver pairs. The research of IC was initiated by Shannon~\cite{shannon} and the channel was first thoroughly studied by Ahlswede~\cite{ahlswede}. Later, Carleial~\cite{carleial} established an improved achievable rate region by applying the superposition coding scheme. In~\cite{han_kobayashi}, Han and Kobayashi obtained the best achievable rate region known to date for the general IC by utilizing simultaneous decoding at the receivers. Recently, this rate region has been re-characterized with superposition encoding for the sub-messages~\cite{cmg_region,new_fourier}. However, the capacity region of the general IC is still an open problem except for several special cases~\cite{han_kobayashi,carleial_2,sato}.

Many variations of the interference channel have also been studied, including the IC with feedback~\cite{jinhua} and the IC with conferencing encoders/decoders~\cite{caoyi}. In this paper, we study another variation of the IC: the state-dependent two-user IC with state information non-causally known at both transmitters. This situation may arise in a multi-cell downlink communication problem, where two interested cells are interfering with each other and the mobiles suffer from some common interference (which can be from other cells and viewed as state) non-causally known at both base-stations. Notably, communication over state-dependent channels has drawn lots of attentions due to its wide applications such as information embedding~\cite{info_embedding} and computer memories with defects~\cite{gamal}. The corresponding framework was also initiated by Shannon in~\cite{shannon_2}, which established the capacity of a state-dependent discrete memoryless (DM) point-to-point channel with causal state information at the transmitter. In~\cite{gelfand}, Gel'fand and Pinsker obtained the capacity for such a point-to-point case with the state information non-causally known at the transmitter. Subsequently, Costa~\cite{costa_dpc} extended Gel'fand-Pinsker coding to the state-dependent additive white Gaussian noise (AWGN) channel, where the state is an additive zero-mean Gaussian interference. This result is known as the dirty-paper coding technique, which achieves the capacity as if there is no such an interference. For the multi-user case, extensions of the afore-mentioned schemes were provided in~\cite{gelfand_2, kim, dpc_bc} for the multiple access channel, the broadcast channel, and the degraded Gaussian relay channel, respectively.

In this paper, we study the DM state-dependent IC with state information non-causally known at the transmitters and develop two coding schemes, both of which jointly apply rate splitting and Gel'fand-Pinsker coding. In the first coding scheme, we deploy simultaneous encoding for the sub-messages and in the second one, we deploy superposition encoding for the sub-messages. The associated achievable rate regions are derived based on the respective coding schemes. 

The rest of the paper is organized as follows. The channel model and the definition of achievable rate region are presented in Section \ref{sec_2}. In Section \ref{sec_3}, we provide two achievable rate regions based on the two different coding schemes, respectively. 
Finally, we conclude the paper in Section \ref{sec_5}.
\section{Channel Model} \label{sec_2}
Consider the interference channel as shown in Fig. 1, where two transmitters communicate with the corresponding receivers through a common channel dependent on state $S$. The transmitters do not cooperate with each other; however, they both know the state information $S$ non-causally, which is unknown to either of the receivers. Each receiver needs to decode the information from the respective transmitter.
\subsection{Notations}
We use the following notations throughout this paper. The random variable is defined as $X$ with value $x$ in a finite set $\mathcal{X}$. Let $p_{X}(x)$ be the probability mass function of $X$ on $\mathcal{X}$. The corresponding sequences are denoted by $x^n$ with length $n$.
\begin{figure}[!t]
\centering
\includegraphics[width=3.5in]{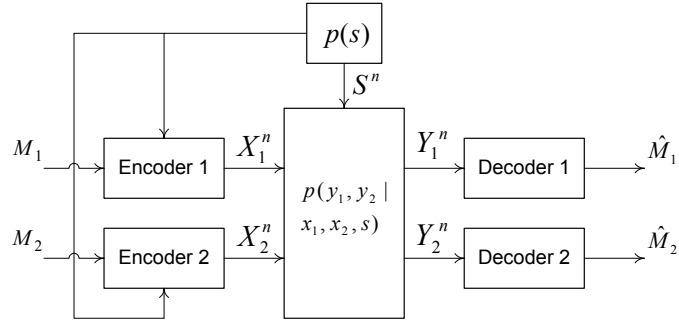}
\caption{The interference channel with state information non-causally known at both transmitters} \label{fig_relay_channel_model}
\end{figure}
\subsection{Discrete Memoryless Case}\label{sec_2_1}
The state-dependent two-user interference channel is defined by $(\mathcal{X}_1,\mathcal{X}_2,\mathcal{Y}_1,\mathcal{Y}_2,\mathcal{S},p(y_1,y_2|x_1,x_2,s))$, where $\mathcal{X}_1,\mathcal{X}_2$ are two input alphabet sets, $\mathcal{Y}_1,\mathcal{Y}_2$ are the corresponding output alphabet sets, $\mathcal{S}$ is the state alphabet set, and $p(y_1,y_2|x_1,x_2,s)$ is the conditional probability of $(y_1,y_2) \in \mathcal{Y}_1 \mathcal{\times} \mathcal{Y}_2$ given $(x_1,x_2,s)\in \mathcal{X}_1 \mathcal{\times}\mathcal{X}_2 \mathcal{\times}\mathcal{S}$. The channel is assumed to be memoryless, i.e.,
\[
p(y_1^n,y_2^n|x_1^n,x_2^n,s^n) = \prod_{i=1}^{n} p(y_{1i},y_{2i}|x_{1i},x_{2i},s_{i}),
\]
where $i$ is the element index for each sequence.

A $(2^{nR_1},2^{nR_2},n)$ code for the above channel consists of two independent message sets $\{1,2,\cdots,2^{nR_{1}}\}$ and $\{1,2,\cdots,2^{nR_{2}}\}$, two encoders that assign a codeword to each message $m_1 \in \{1,2,\cdots,2^{nR_{1}}\}$ and $m_2 \in \{1,2,\cdots,2^{nR_{2}}\}$ based on the non-causally known state information $s^n$, and two decoders that determine the estimated messages $\hat{m}_1$ and $\hat{m}_2$ or declare an error from the received sequences.

The average probability of error is defined as:
\begin{equation}
P_e^{(n)} = \frac{1}{2^{n(R_1+R_2)}}\sum_{m_1,m_2} \textrm{Pr}\{\hat{m}_1 \neq m_1 \textrm{ or } \hat{m}_2 \neq m_2|(m_1,m_2) \textrm{ is sent}\},
\end{equation}
where $(m_1,m_2)$ is assumed to be uniformly distributed in $\{1,2,\cdots,2^{nR_{1}}\}\times\{1,2,\cdots,2^{nR_{2}}\}$.

\begin{Definition}
A rate pair $(R_1,R_2)$ of non-negative real values is achievable if there exists a sequence of $(2^{nR_1},2^{nR_2},n)$ codes with $P_e^{(n)}\to 0$ as $n \to \infty$. The set of all achievable rate pairs is defined as the capacity region.
\end{Definition}
\section{Achievable Rate Regions for the DM Interference Channel with State Information} \label{sec_3}
In this section, we propose two new coding schemes for the DM interference channel with state information non-causally known at both transmitters and present the associated achievable rate regions. For both coding schemes, we jointly deploy rate splitting and Gel'fand-Pinsker coding. In the first coding scheme, we use simultaneous encoding on the sub-messages, while in the second one we apply superposition encoding.
\subsection{Simultaneous Encoding}\label{sec_3_1}
Now we introduce the following rate region achieved by the first coding scheme, which combines rate splitting and Gel'fand-Pinsker coding.
\begin{Theorem}\label{theorem_1}
For a fixed probability distribution $p(q)p(u_1|q,s)p(v_1|q,s)p(u_2|q,s)p(v_2|q,s)$, let $\mathcal{R}_1$ be the set of all non-negative rate tuple $(R_{10},R_{11},R_{20},R_{22})$ satisfying
{\small
\begin{eqnarray}
\label{eq_rate_constraint_11} R_{11} &\leq& I(U_1;U_2|Q)+I(U_1,U_2;V_1|Q)+I(V_1;Y_1|U_1,U_2,Q)-I(V_1;S|Q),\\
\label{eq_rate_constraint_12} R_{10} &\leq& I(U_1;U_2|Q)+I(U_1,U_2;V_1|Q)+I(U_1;Y_1|V_1,U_2,Q)-I(U_1;S|Q),\\
\label{eq_rate_constraint_13} R_{10}+R_{11} &\leq& I(U_1;U_2|Q)+I(U_1,U_2;V_1|Q)+I(U_1,V_1;Y_1|U_2,Q)-I(U_1;S|Q)-I(V_1;S|Q),\\
\label{eq_rate_constraint_14} R_{11}+R_{20} &\leq& I(U_1;U_2|Q)+I(U_1,U_2;V_1|Q)+I(V_1,U_2;Y_1|U_1,Q)-I(V_1;S|Q)-I(U_2;S|Q),\\
\label{eq_rate_constraint_15} R_{10}+R_{20} &\leq& I(U_1;U_2|Q)+I(U_1,U_2;V_1|Q)+I(U_1,U_2;Y_1|V_1,Q)-I(U_1;S|Q)-I(U_2;S|Q),\\
\label{eq_rate_constraint_16} R_{10}+R_{11}+R_{20} &\leq& I(U_1;U_2|Q)+I(U_1,U_2;V_1|Q)+I(U_1,V_1,U_2;Y_1|Q)-I(U_1;S|Q)-I(V_1;S|Q)-I(U_2;S|Q),\\
\label{eq_rate_constraint_21} R_{22} &\leq& I(U_2;U_1|Q)+I(U_2,U_1;V_2|Q)+I(V_2;Y_2|U_2,U_1,Q)-I(V_2;S|Q),\\
\label{eq_rate_constraint_22} R_{20} &\leq& I(U_2;U_1|Q)+I(U_2,U_1;V_2|Q)+I(U_2;Y_2|V_2,U_1,Q)-I(U_2;S|Q),\\
\label{eq_rate_constraint_23} R_{20}+R_{22} &\leq& I(U_2;U_1|Q)+I(U_2,U_1;V_2|Q)+I(U_2,V_2;Y_2|U_1,Q)-I(U_2;S|Q)-I(V_2;S|Q),\\
\label{eq_rate_constraint_24} R_{22}+R_{10} &\leq& I(U_2;U_1|Q)+I(U_2,U_1;V_2|Q)+I(V_2,U_1;Y_2|U_2,Q)-I(V_2;S|Q)-I(U_1;S|Q),\\
\label{eq_rate_constraint_25} R_{20}+R_{10} &\leq& I(U_2;U_1|Q)+I(U_2,U_1;V_2|Q)+I(U_2,U_1;Y_2|V_2,Q)-I(U_2;S|Q)-I(U_1;S|Q),\\
\label{eq_rate_constraint_26} R_{20}+R_{22}+R_{10} &\leq& I(U_2;U_1|Q)+I(U_2,U_1;V_2|Q)+I(U_2,V_2,U_1;Y_2|Q)-I(U_2;S|Q)-I(V_2;S|Q)-I(U_1;S|Q).
\end{eqnarray}}
Then for any $(R_{10},R_{11},R_{20},R_{22}) \in \mathcal{R}_1$, the rate pair $(R_{10}+R_{11},R_{20}+R_{22})$ is achievable for the DM interference channel with state information non-causally known at both transmitters.
\end{Theorem}

\begin{proof}
In the achievable coding scheme for Theorem \ref{theorem_1}, the message at the $j$th transmitter is splitted into two parts: the public message $m_{j0}$ and the private message $m_{jj}$. Subsequently, the $j$th decoder tries to decode the corresponding messages from the intending transmitter and the public message of the interfering transmitter. Furthermore, Gel'fand-Pinsker coding is utilized to help both transmitters send the messages with the non-causal knowledge of the state information. Here we presume that the message pairs are chosen uniformly on the message sets for both transmitters.

Codebook generation: Fix the probability distribution $p(q)p(u_1|q,s)p(v_1|q,s)p(u_2|q,s)p(v_2|q,s)$. Also define the following function for the $j$th user that maps $\mathcal{U}_j\mathcal{\times V}_j\mathcal{\times S}$ to $\mathcal{X}_j$:
\[
x_{ji}=F_j(u_{ji},v_{ji},s_i),
\]
where $i$ is the element index of each sequence.

Generate the time-sharing sequence $q^n \sim \prod_{i=1}^{n} p_{Q}(q_i)$. For the $j$th user, $u_{j}^{n}(m_{j0},l_{j0})$ is randomly and conditionally independently generated according to $\prod_{i=1}^{n}p_{U_j|Q}(u_{ji}|q_i)$, for $m_{j0} \in \{1,2,\cdots,2^{nR_{j0}}\}$ and $l_{j0} \in \{1,2,\cdots,2^{nR_{j0}'}\}$. Similarly, $v_{j}^{n}(m_{jj},l_{jj})$ is randomly and conditionally independently generated according to $\prod_{i=1}^{n}p_{V_j|Q}(v_{ji}|q_i)$, for $m_{jj} \in \{1,2,\cdots,2^{nR_{jj}}\}$ and $l_{jj} \in \{1,2,\cdots,2^{nR_{jj}'}\}$.

Encoding: To send the message $m_j=(m_{j0},m_{jj})$, the $j$th encoder first tries to find the pair $(l_{j0},l_{jj})$ such that the following joint typicality holds: $(q^n,u_j^n(m_{j0},l_{j0}),s^n) \in T_\epsilon^{(n)}$ and $(q^n,v_j^n(m_{jj},l_{jj}),s^n) \in T_\epsilon^{(n)}$. If successful, $(q^n,u_j^n(m_{j0},l_{j0}),v_j^n(m_{jj},l_{jj}),s^n)$ is also jointly typical with high probability, and the $j$th encoder sends $x_j$ where the $i$th element is $x_{ji} = F_j(u_{ji}(m_{j0},l_{j0}),v_{ji}(m_{jj},l_{jj}),s_{i})$. If not, the $j$th encoder transmits $x_j$ where the $i$th element is $x_{ji} = F_j(u_{ji}(m_{j0},1),v_{ji}(m_{jj},1),s_{i})$.

Decoding: Decoder $1$ finds the unique message pair $(\hat{m}_{10},\hat{m}_{11})$ such that $(q^n,u_1^n(\hat{m}_{10},\hat{l}_{10}),u_2^n(\hat{m}_{20},\hat{l}_{20}),\\v_{1}^{n}(\hat{m}_{11},\hat{l}_{11}),y_1^n) \in T_\epsilon^{(n)}$ for some $\hat{l}_{10} \in \{1,2,\cdots,2^{nR_{10}'}\}$, $\hat{m}_{20} \in \{1,2,\cdots,2^{nR_{20}}\}$, $\hat{l}_{20}\in\{1,2,\cdots,2^{nR_{20}'}\}$, and $\hat{l}_{11}\in \{1,2,\cdots,2^{nR_{11}'}\}$. If no such unique pair exists, the decoder declares an error. Decoder $2$ determines the unique message pair $(\hat{m}_{20},\hat{m}_{22})$ in a similar way.

Analysis of probability of error: Here the probability of error is the same for each message pair since the transmitted message pair is chosen with a uniform distribution on the message set. Without loss of generality, we assume $(1,1)$ for user $1$ and $(1,1)$ for user $2$ are sent over the channel. First we consider the encoding error probability at transmitter $1$. Define the following error events:
\[
\xi_1 = \left\{\left(q^n,u_1^n\left(1,l_{10}\right),s^n\right) \notin T_\epsilon^{(n)} \textrm{ for all } l_{10} \in \{1,2,\cdots,2^{nR_{10}'}\}\right\},
\]
\[
\xi_2 = \left\{\left(q^n,v_1^n\left(1,l_{11}\right),s^n\right) \notin T_\epsilon^{(n)} \textrm{ for all } l_{11} \in \{1,2,\cdots,2^{nR_{11}'}\}\right\}.
\]

The probability of the error event $\xi_1$ can be bounded as follows:
\begin{eqnarray*}
P(\xi_1) &=& \prod_{l_{10}=1}^{2^{nR_{10}'}} \left(1-P\left(\left\{\left(q^n,u_1^n\left(1,l_{10}\right),s^n\right) \in T_\epsilon^{(n)}\right\}\right)\right)\\
&\leq& \left(1-2^{-n\left(I(U_1;S|Q)+\delta_1(\epsilon)\right)}\right)^{2^{nR_{10}'}}\\
&\leq& e^{-2^{n\left(R_{10}'-I(U_1;S|Q)+\delta_1(\epsilon)\right)}},
\end{eqnarray*}
where $\delta_1(\epsilon)\to 0$ as $\epsilon \to 0$. Therefore, the probability of $\xi_1$ goes to $0$ as $n\to\infty$ if
\begin{equation}\label{eq_prob_xi1}
R_{10}' \geq I(U_1;S|Q).
\end{equation}
Similarly, the probability of $\xi_2$ can also be upper bounded by an arbitrarily small number as $n\to\infty$ if
\begin{equation}\label{eq_prob_xi2}
R_{11}' \geq I(V_1;S|Q).
\end{equation}

The encoding error probability at transmitter $1$ can be calculated as:
\[
P_{\textrm{enc}1} = P\left(\xi_1 \cup \xi_2\right) \leq P(\xi_1) + P(\xi_2),
\]
which goes to $0$ as $n\to \infty$ if \eqref{eq_prob_xi1} and \eqref{eq_prob_xi2} are satisfied.

Now we consider the error analysis at the decoder $1$. Denote the right Gel'fand-Pinsker coding indices chosen by the encoders as $(L_{10},L_{11})$ and $(L_{20},L_{22})$. Define the following error events:
{\small
\begin{eqnarray*}
\xi_{31} &=&  \left\{\left(q^n,u_1^n\left(1,L_{10}\right),u_2^n\left(1,L_{20}\right),v_1^n\left(m_{11},l_{11}\right),y_1^n\right) \in T_\epsilon^{(n)} \textrm{ for } m_{11} \neq 1, \textrm{ and some } l_{11}\right\},\\
\xi_{32} &=&  \left\{\left(q^n,u_1^n\left(1,L_{10}\right),u_2^n\left(1,l_{20}\right),v_1^n\left(m_{11},l_{11}\right),y_1^n\right) \in T_\epsilon^{(n)} \textrm{ for } m_{11} \neq 1, \textrm{ and some } l_{11},l_{20}\neq L_{20}\right\},\\
\xi_{33} &=&  \left\{\left(q^n,u_1^n\left(1,l_{10}\right),u_2^n\left(1,L_{20}\right),v_1^n\left(m_{11},l_{11}\right),y_1^n\right) \in T_\epsilon^{(n)} \textrm{ for } m_{11} \neq 1, \textrm{ and some } l_{11},l_{10}\neq L_{10}\right\},\\
\xi_{34} &=&  \left\{\left(q^n,u_1^n\left(1,l_{10}\right),u_2^n\left(1,l_{20}\right),v_1^n\left(m_{11},l_{11}\right),y_1^n\right) \in T_\epsilon^{(n)} \textrm{ for } m_{11} \neq 1, \textrm{ and some } l_{11},l_{10}\neq L_{10},l_{20}\neq L_{20}\right\},\\
\xi_{41} &=& \left\{\left(q^n,u_1^n\left(m_{10},l_{10}\right),u_2^n\left(1,L_{20}\right),v_1^n\left(1,L_{11}\right),y_1^n\right) \in T_\epsilon^{(n)} \textrm{ for } m_{10} \neq 1, \textrm{ and some } l_{10}\right\},\\
\xi_{42} &=& \left\{\left(q^n,u_1^n\left(m_{10},l_{10}\right),u_2^n\left(1,l_{20}\right),v_1^n\left(1,L_{11}\right),y_1^n\right) \in T_\epsilon^{(n)} \textrm{ for } m_{10} \neq 1, \textrm{ and some } l_{10},l_{20}\neq L_{20}\right\},\\
\xi_{43} &=& \left\{\left(q^n,u_1^n\left(m_{10},l_{10}\right),u_2^n\left(1,L_{20}\right),v_1^n\left(1,l_{11}\right),y_1^n\right) \in T_\epsilon^{(n)} \textrm{ for } m_{10} \neq 1, \textrm{ and some } l_{10},l_{11}\neq L_{11}\right\},\\
\xi_{44} &=& \left\{\left(q^n,u_1^n\left(m_{10},l_{10}\right),u_2^n\left(1,l_{20}\right),v_1^n\left(1,l_{11}\right),y_1^n\right) \in T_\epsilon^{(n)} \textrm{ for } m_{10} \neq 1, \textrm{ and some } l_{10},l_{20}\neq L_{20},l_{11}\neq L_{11}\right\},\\
\xi_{51} &=& \left\{\left(q^n,u_1^n\left(m_{10},l_{10}\right),u_2^n\left(1,L_{20}\right),v_1^n\left(m_{11},l_{11}\right),y_1^n\right) \in T_\epsilon^{(n)} \textrm{ for } m_{10} \neq 1,\ m_{11} \neq 1, \textrm{ and some } l_{10},l_{11}\right\},\\
\xi_{52} &=& \left\{\left(q^n,u_1^n\left(m_{10},l_{10}\right),u_2^n\left(1,l_{20}\right),v_1^n\left(m_{11},l_{11}\right),y_1^n\right) \in T_\epsilon^{(n)} \textrm{ for } m_{10} \neq 1,\ m_{11} \neq 1, \textrm{ and some } l_{10},l_{11},l_{20}\neq L_{20}\right\},\\
\xi_{61} &=& \left\{\left(q^n,u_1^n\left(1,L_{10}\right),u_2^n\left(m_{20},l_{20}\right),v_1^n\left(m_{11},l_{11}\right),y_1^n\right) \in T_\epsilon^{(n)} \textrm{ for } m_{20} \neq 1,\ m_{11} \neq 1, \textrm{ and some } l_{20},l_{11}\right\},\\
\xi_{62} &=& \left\{\left(q^n,u_1^n\left(1,l_{10}\right),u_2^n\left(m_{20},l_{20}\right),v_1^n\left(m_{11},l_{11}\right),y_1^n\right) \in T_\epsilon^{(n)} \textrm{ for } m_{20} \neq 1,\ m_{11} \neq 1, \textrm{ and some } l_{20},l_{11},l_{10}\neq L_{10}\right\},\\
\xi_{71} &=& \left\{\left(q^n,u_1^n\left(m_{10},l_{10}\right),u_2^n\left(m_{20},l_{20}\right),v_1^n\left(1,L_{11}\right),y_1^n\right) \in T_\epsilon^{(n)} \textrm{ for } m_{10} \neq 1,\ m_{20} \neq 1, \textrm{ and some } l_{10},l_{20}\right\},\\
\xi_{72} &=& \left\{\left(q^n,u_1^n\left(m_{10},l_{10}\right),u_2^n\left(m_{20},l_{20}\right),v_1^n\left(1,l_{11}\right),y_1^n\right) \in T_\epsilon^{(n)} \textrm{ for } m_{10} \neq 1,\ m_{20} \neq 1, \textrm{ and some } l_{10},l_{20},l_{11}\neq L_{11}\right\},\\
\xi_8 &=& \big\{\left(q^n,u_1^n\left(m_{10},l_{10}\right),u_2^n\left(m_{20},l_{20}\right),v_1^n\left(m_{11},l_{11}\right),y_1^n\right) \in T_\epsilon^{(n)} \textrm{ for } m_{10} \neq 1,\ m_{20} \neq 1,\ m_{11}\neq 1, \\ && \textrm{ and some } l_{10},l_{20},l_{11}\big\}.
\end{eqnarray*}
}
The probability of $\xi_{31}$ can be bounded as follows:
\begin{eqnarray*}
P(\xi_{31}) &=& \sum_{m_{11}=2}^{2^{nR_{11}}}\ \sum_{l_{11}=1}^{2^{R_{11}'}} P\left(\{\left(q^n,u_1^n\left(1,L_{10}\right),u_2^n\left(1,L_{20}\right),v_1^n\left(m_{11},l_{11}\right),y_1^n\right) \in T_\epsilon^{(n)}\}\right)\\
&\leq& 2^{n\left(R_{11}+R_{11}'\right)} \sum_{(q^n,u_1^n,u_2^n,v_1^n,y_1^n)\in T_\epsilon^{(n)}} p(q^n)p(u_1^n|q^n)p(u_2^n|q^n)p(v_1^n|q^n)p(y_1^n|u_1^n,u_2^n,q^n)\\
&\leq& 2^{n\left(R_{11}+R_{11}'\right)} 2^{-n\left(H(Q)+H(U_1|Q)+H(U_2|Q)+H(V_1|Q)+H(Y_1|U_1,U_2,Q)-H(Q,U_1,U_2,V_1,Y_1)-\delta_2(\epsilon)\right)}\\
&\leq& 2^{n\left(R_{11}+R_{11}'\right)} 2^{-n\left(I(U_1;U_2|Q)+I(U_1,U_2;V_1|Q)+I(V_1;Y_1|U_1,U_2,Q)-\delta_2(\epsilon)\right)},
\end{eqnarray*}
where $\delta_2(\epsilon)\to 0$ as $\epsilon \to 0$. Obviously, the probability that $\xi_{31}$ happens goes to $0$ if
\begin{equation}\label{eq_prob_xi31}
R_{11}+R_{11}' \leq I(U_1;U_2|Q)+I(U_1,U_2;V_1|Q)+I(V_1;Y_1|U_1,U_2,Q).
\end{equation}
Similarly, the error probability corresponding to the left error events goes to $0$, respectively, if
\begin{eqnarray}
\label{eq_prob_xi32} R_{11}+R_{11}'+R_{20}' &\leq& I(U_1;U_2|Q)+I(U_1,U_2;V_1|Q)+I(V_1,U_2;Y_1|U_1,Q),\\
\label{eq_prob_xi33} R_{11}+R_{10}'+R_{11}' &\leq& I(U_1;U_2|Q)+I(U_1,U_2;V_1|Q)+I(U_1,V_1;Y_1|U_2,Q),\\
\label{eq_prob_xi34} R_{11}+R_{10}'+R_{11}'+R_{20}' &\leq& I(U_1;U_2|Q)+I(U_1,U_2;V_1|Q)+I(U_1,V_1,U_2;Y_1|Q),\\
\label{eq_prob_xi41} R_{10}+R_{10}' &\leq& I(U_1;U_2|Q)+I(U_1,U_2;V_1|Q)+I(U_1;Y_1|V_1,U_2,Q),\\
\label{eq_prob_xi42} R_{10}+R_{10}'+R_{20}' &\leq& I(U_1;U_2|Q)+I(U_1,U_2;V_1|Q)+I(U_1,U_2;Y_1|V_1,Q),\\
\label{eq_prob_xi43} R_{10}+R_{10}'+R_{11}' &\leq& I(U_1;U_2|Q)+I(U_1,U_2;V_1|Q)+I(U_1,V_1;Y_1|U_2,Q),\\
\label{eq_prob_xi44} R_{10}+R_{10}'+R_{11}'+R_{20}' &\leq& I(U_1;U_2|Q)+I(U_1,U_2;V_1|Q)+I(U_1,V_1,U_2;Y_1|Q),\\
\label{eq_prob_xi51} R_{10}+R_{11}+R_{10}'+R_{11}' &\leq& I(U_1;U_2|Q)+I(U_1,U_2;V_1|Q)+I(U_1,V_1;Y_1|U_2,Q), \\
\label{eq_prob_xi52} R_{10}+R_{11}+R_{10}'+R_{11}'+R_{20}' &\leq& I(U_1;U_2|Q)+I(U_1,U_2;V_1|Q)+I(U_1,V_1,U_2;Y_1|Q), \\
\label{eq_prob_xi61} R_{11}+R_{20}+R_{11}'+R_{20}' &\leq& I(U_1;U_2|Q)+I(U_1,U_2;V_1|Q)+I(V_1,U_2;Y_1|U_1,Q), \\
\label{eq_prob_xi62} R_{11}+R_{20}+R_{10}'+R_{11}'+R_{20}' &\leq& I(U_1;U_2|Q)+I(U_1,U_2;V_1|Q)+I(U_1,V_1,U_2;Y_1|Q), \\
\label{eq_prob_xi71} R_{10}+R_{20}+R_{10}'+R_{20}' &\leq& I(U_1;U_2|Q)+I(U_1,U_2;V_1|Q)+I(U_1,U_2;Y_1|V_1,Q), \\
\label{eq_prob_xi72} R_{10}+R_{20}+R_{10}'+R_{11}'+R_{20}' &\leq& I(U_1;U_2|Q)+I(U_1,U_2;V_1|Q)+I(U_1,V_1,U_2;Y_1|Q), \\
\label{eq_prob_xi8} R_{10}+R_{11}+R_{20}+R_{10}'+R_{11}'+R_{20}' &\leq& I(U_1;U_2|Q)+I(U_1,U_2;V_1|Q)+I(U_1,V_1,U_2;Y_1|Q).
\end{eqnarray}
Note that there are some redundant inequalities in \eqref{eq_prob_xi31}-\eqref{eq_prob_xi8}: \eqref{eq_prob_xi32} is implied by \eqref{eq_prob_xi61}; \eqref{eq_prob_xi33} is implied by \eqref{eq_prob_xi51}; \eqref{eq_prob_xi42} is implied by \eqref{eq_prob_xi71}; \eqref{eq_prob_xi43} is implied by \eqref{eq_prob_xi51}; \eqref{eq_prob_xi34}, \eqref{eq_prob_xi44}, \eqref{eq_prob_xi52}, \eqref{eq_prob_xi62}, and \eqref{eq_prob_xi72} are implied by \eqref{eq_prob_xi8}. By combining with the error analysis at the encoder, we can recast the rate constraints \eqref{eq_prob_xi31}-\eqref{eq_prob_xi8} as:
\begin{eqnarray*}
R_{11} &\leq& I(U_1;U_2|Q)+I(U_1,U_2;V_1|Q)+I(V_1;Y_1|U_1,U_2,Q)-I(V_1;S|Q),\\
R_{10} &\leq& I(U_1;U_2|Q)+I(U_1,U_2;V_1|Q)+I(U_1;Y_1|V_1,U_2,Q)-I(U_1;S|Q),\\
R_{10}+R_{11} &\leq& I(U_1;U_2|Q)+I(U_1,U_2;V_1|Q)+I(U_1,V_1;Y_1|U_2,Q)-I(U_1;S|Q)-I(V_1;S|Q),\\
R_{11}+R_{20} &\leq& I(U_1;U_2|Q)+I(U_1,U_2;V_1|Q)+I(V_1,U_2;Y_1|U_1,Q)-I(V_1;S|Q)-I(U_2;S|Q),\\
R_{10}+R_{20} &\leq& I(U_1;U_2|Q)+I(U_1,U_2;V_1|Q)+I(U_1,U_2;Y_1|V_1,Q)-I(U_1;S|Q)-I(U_2;S|Q),\\
R_{10}+R_{11}+R_{20} &\leq& I(U_1;U_2|Q)+I(U_1,U_2;V_1|Q)+I(U_1,V_1,U_2;Y_1|Q)-I(U_1;S|Q)-I(V_1;S|Q)-I(U_2;S|Q).
\end{eqnarray*}

The error analysis for transmitter $2$ and decoder $2$ is similar to user $1$ and is omitted here. Correspondingly, \eqref{eq_rate_constraint_21} to \eqref{eq_rate_constraint_26} show the rate constraints for user $2$.
In addition, the right hand sides of the inequalities \eqref{eq_rate_constraint_11} to \eqref{eq_rate_constraint_26} are guaranteed to be non-negative when choosing the probability distribution. As long as \eqref{eq_rate_constraint_11} to \eqref{eq_rate_constraint_26} are satisfied, the probability of error can be bounded by the sum of the error probability at the encoders and the decoders, which goes to $0$ as $n\to\infty$.
\end{proof}

An explicit description of the achievable rate region can be obtained by applying Fourier-Motzkin algorithm on our implicit description \eqref{eq_rate_constraint_11}-\eqref{eq_rate_constraint_26}. We omit it here due to its high complexity and the space limitation.
\subsection{Superposition Encoding}\label{sec_3_2}
We now present another coding scheme, which applies superposition encoding for the sub-messages. The achievable rate region is given in the following theorem.
\begin{Theorem}\label{theorem_2}
For a fixed probability distribution $p(q)p(u_1|s,q)p(v_1|u_1,s,q)p(u_2|s,q)p(v_2|u_2,s,q)$, let $\mathcal{R}_2$ be the set of all non-negative rate tuple $(R_{10},R_{11},R_{20},R_{22})$ satisfying
\begin{eqnarray}
\label{eq_rate_constraint_2_11} R_{11} &\leq& I(U_1,V_1;U_2|Q)+I(V_1;Y_1|U_1,U_2,Q)-I(V_1;S|U_1,Q),\\
\label{eq_rate_constraint_2_12} R_{10}+R_{11} &\leq& I(U_1,V_1;U_2|Q)+I(U_1,V_1;Y_1|U_2,Q)-I(U_1,V_1;S|Q),\\
\label{eq_rate_constraint_2_13} R_{11}+R_{20} &\leq& I(U_1,V_1;U_2|Q)+I(V_1,U_2;Y_1|U_1,Q)-I(V_1;S|U_1,Q)-I(U_2;S|Q),\\
\label{eq_rate_constraint_2_14} R_{10}+R_{11}+R_{20} &\leq& I(U_1,V_1;U_2|Q)+I(U_1,V_1,U_2;Y_1|Q)-I(U_1,V_1;S|Q)-I(U_2;S|Q),\\
\label{eq_rate_constraint_2_21} R_{22} &\leq& I(U_2,V_2;U_1|Q)+I(V_2;Y_2|U_2,U_1,Q)-I(V_2;S|U_2,Q),\\
\label{eq_rate_constraint_2_22} R_{20}+R_{22} &\leq& I(U_2,V_2;U_1|Q)+I(U_2,V_2;Y_2|U_1,Q)-I(U_2,V_2;S|Q),\\
\label{eq_rate_constraint_2_23} R_{22}+R_{10} &\leq& I(U_2,V_2;U_1|Q)+I(V_2,U_1;Y_2|U_2,Q)-I(V_2;S|U_2,Q)-I(U_1;S|Q),\\
\label{eq_rate_constraint_2_24} R_{20}+R_{22}+R_{10} &\leq& I(U_2,V_2;U_1|Q)+I(U_2,V_2,U_1;Y_2|Q)-I(U_2,V_2;S|Q)-I(U_1;S|Q).
\end{eqnarray}
Then for any $(R_{10},R_{11},R_{20},R_{22}) \in  \mathcal{R}_2$, the rate pair $(R_{10}+R_{11},R_{20}+R_{22})$ is achievable for the DM interference channel defined in Section \ref{sec_2}.
\end{Theorem}
\begin{proof}
Compared with the first coding scheme, the rate splitting structure is also applied in the achievable scheme of Theorem \ref{theorem_2}. The main difference here is that instead of simultaneous encoding, now the private message $m_{jj}$ is superimposed on the public message $m_{j0}$ for the $j$th transmitter. Gel'fand-Pinsker coding is also utilized to help the transmitters send both public and private messages.

Codebook generation: Fix the probability distribution $p(q)p(u_1|s,q)p(v_1|u_1,s,q)p(u_2|s,q)p(v_2|u_2,s,q)$. First generate the time-sharing sequence $q^n\sim\prod_{i=1}^{n} p_{Q}(q_i)$. For the $j$th user, $u_j^{n}(m_{j0},l_{j0})$ is randomly and conditionally independently generated according to $\prod_{i=1}^{n}p_{U_j|Q}(u_{ji}|q_i)$, for $m_{j0} \in \{1,2,\cdots,2^{nR_{j0}}\}$ and $l_{j0}\in \{1,2,\cdots,2^{nR_{j0}'}\}$. For each $u_j^n(m_{j0},l_{j0})$, $v_{j}^{n}(m_{j0},l_{j0},m_{jj},l_{jj})$ is randomly and conditionally independently generated according to $\prod_{i=1}^{n}p_{V_j|U_j,Q}(v_{ji}|u_{ji},q_{i})$, for $m_{jj} \in \{1,2,\cdots,2^{nR_{jj}}\}$ and $l_{jj} \in \{1,2,\cdots,2^{nR_{jj}'}\}$.

Encoding: To send the message $m_j=(m_{j0},m_{jj})$, the $j$th encoder first tries to find $l_{j0}$ such that $(q^n,u_j^n(m_{j0},l_{j0}),s^n) \in T_\epsilon^{(n)}$ holds. Then for this specific $l_{j0}$, find $l_{jj}$ such that $(q^n,u_j^n(m_{j0},l_{j0}),v_j^n(m_{j0},l_{j0},m_{jj},l_{jj}),s^n) \in T_\epsilon^{(n)}$ holds. If successful, the $j$th encoder sends $v_{j}^{n}(m_{j0},l_{j0},m_{jj},l_{jj})$. If not, the $j$th encoder transmits $v_{j}^{n}(m_{j0},1,m_{jj},1)$.

Decoding: Decoder $1$ finds the unique message pair $(\hat{m}_{10},\hat{m}_{11})$ such that $(q^n,u_1^n(\hat{m}_{10},\hat{l}_{10}),u_2^n(\hat{m}_{20},\hat{l}_{20}),\\ v_1^n(\hat{m}_{10},\hat{l}_{10},\hat{m}_{11},\hat{l}_{11}),y_1^n) \in T_{\epsilon}^{(n)}$ for some $\hat{l}_{10} \in \{1,2,\cdots,2^{nR_{10}'}\}$, $\hat{m}_{20} \in \{1,2,\cdots,2^{nR_{20}}\}$,$\hat{l}_{20} \in \{1,2,\cdots,2^{nR_{20}'}\}$, and $\hat{l}_{11} \in \{1,2,\cdots,2^{nR_{11}'}\}$. If no such unique pair exists, the decoder declares an error. Decoder $2$ determines the unique message pair $(\hat{m}_{20},\hat{m}_{22})$ similarly.

Analysis of probability of error: Similar to the proof in Theorem \ref{theorem_1}, we assume message $(1,1)$ and $(1,1)$ are sent for both transmitters. First we consider the encoding error probability at transmitter $1$. Define the following error events:
\[
\xi_1' = \left\{\left(q^n,u_1^n\left(1,l_{10}\right),s^n\right) \notin T_\epsilon^{(n)} \textrm{ for all } l_{10} \in \{1,2,\cdots,2^{nR_{10}'}\}\right\},
\]
\[
\xi_2' = \left\{\left(q^n,u_1^n(m_{10},l_{10}),v_1^n\left(1,l_{10},1,l_{11}\right),s^n\right) \notin T_\epsilon^{(n)} \textrm{ for all } l_{11} \in \{1,2,\cdots,2^{nR_{11}'}\} \textrm{ and previously found typical } l_{10}\big|\bar{\xi}_1'\right\}.
\]

The probability of the error event $\xi_1'$ can be bounded as follows:
\begin{eqnarray*}
P(\xi_1') &=& \prod_{l_{10}=1}^{2^{nR_{10}'}} \left(1-P\left(\left\{\left(q^n,u_1^n\left(1,l_{10}\right),s^n\right) \in T_\epsilon^{(n)}\right\}\right)\right)\\
&\leq& \left(1-2^{-n\left(I(U_1;S|Q)+\delta_1'(\epsilon)\right)}\right)^{2^{nR_{10}'}}\\
&\leq& e^{-2^{n\left(R_{10}'-I(U_1;S|Q)+\delta_1'(\epsilon)\right)}},
\end{eqnarray*}
where $\delta_1'(\epsilon)\to 0$ as $\epsilon \to 0$. Therefore, the probability of $\xi_1'$ goes to $0$ as $n\to\infty$ if
\begin{equation}\label{eq_prob_hatxi1}
R_{10}' \geq I(U_1;S|Q).
\end{equation}
Similarly, for the previously found typical $l_{10}$, the probability of $\xi_2'$ can be upper bounded as follows:
\begin{eqnarray*}
P(\xi_2') &=& \prod_{l_{11}=1}^{2^{nR_{11}'}} \left(1-P\left(\left\{\left(q^n,u_1^n\left(1,l_{10}\right),v_1^n\left(1,l_{10},1,l_{11}\right),s^n\right) \in T_\epsilon^{(n)}\right\}\right)\right)\\
&\leq& \left(1-2^{n\left(H(Q,U_1,V_1,S)-H(Q,U_1,S)-H(V_1|U_1,Q)-\delta_2'(\epsilon)\right)}\right)^{2^{nR_{11}'}}\\
&\leq& \left(1-2^{-n\left(I(V_1;S|U_1,Q)+\delta_2'(\epsilon)\right)}\right)^{2^{nR_{11}'}}\\
&\leq& e^{-2^{n\left(R_{11}'-I(V_1;S|U_1,Q)+\delta_2'(\epsilon)\right)}},
\end{eqnarray*}
where $\delta_2'(\epsilon)\to 0$ as $\epsilon \to 0$. Therefore, the probability of $\xi_2'$ goes to $0$ as $n\to\infty$ if
\begin{equation}\label{eq_prob_hatxi2}
R_{11}' \geq I(V_1;S|U_1,Q).
\end{equation}

The encoding error probability at transmitter $1$ can be calculated as:
\[
P_{\textrm{enc}1} = P(\xi_1') + P(\xi_2'),
\]
which goes to $0$ as $n\to \infty$ if \eqref{eq_prob_hatxi1} and \eqref{eq_prob_hatxi2} are satisfied.

Now we consider the error analysis at the decoder $1$. Denote the right Gel'fand-Pinsker coding indices chosen by the encoders as $(L_{10},L_{11})$ and $(L_{20},L_{22})$. Define the following error events:
{\small
\begin{eqnarray*}
\xi_{31}' &=&  \left\{\left(q^n,u_1^n\left(1,L_{10}\right),u_2^n\left(1,L_{20}\right),v_1^n\left(1,L_{10},m_{11},l_{11}\right),y_1^n\right) \in T_\epsilon^{(n)} \textrm{ for } m_{11} \neq 1, \textrm{ and some } l_{11}\right\},\\
\xi_{32}' &=&  \left\{\left(q^n,u_1^n\left(1,L_{10}\right),u_2^n\left(1,l_{20}\right),v_1^n\left(1,L_{10},m_{11},l_{11}\right),y_1^n\right) \in T_\epsilon^{(n)} \textrm{ for } m_{11} \neq 1, \textrm{ and some } l_{11},l_{20}\neq L_{20}\right\},\\
\xi_{33}' &=&  \left\{\left(q^n,u_1^n\left(1,l_{10}\right),u_2^n\left(1,L_{20}\right),v_1^n\left(1,l_{10},m_{11},l_{11}\right),y_1^n\right) \in T_\epsilon^{(n)} \textrm{ for } m_{11} \neq 1, \textrm{ and some } l_{11},l_{10}\neq L_{10}\right\},\\
\xi_{34}' &=&  \left\{\left(q^n,u_1^n\left(1,l_{10}\right),u_2^n\left(1,l_{20}\right),v_1^n\left(1,l_{10},m_{11},l_{11}\right),y_1^n\right) \in T_\epsilon^{(n)} \textrm{ for } m_{11} \neq 1, \textrm{ and some } l_{11},l_{10}\neq L_{10},l_{20}\neq L_{20}\right\},\\
\xi_{41}' &=& \left\{\left(q^n,u_1^n\left(m_{10},l_{10}\right),u_2^n\left(1,L_{20}\right),v_1^n\left(m_{10},l_{10},1,L_{11}\right),y_1^n\right) \in T_\epsilon^{(n)} \textrm{ for } m_{10} \neq 1, \textrm{ and some } l_{10}\right\},\\
\xi_{42}' &=& \left\{\left(q^n,u_1^n\left(m_{10},l_{10}\right),u_2^n\left(1,l_{20}\right),v_1^n\left(m_{10},l_{10},1,L_{11}\right),y_1^n\right) \in T_\epsilon^{(n)} \textrm{ for } m_{10} \neq 1, \textrm{ and some } l_{10},l_{20}\neq L_{20}\right\},\\
\xi_{43}' &=& \left\{\left(q^n,u_1^n\left(m_{10},l_{10}\right),u_2^n\left(1,L_{20}\right),v_1^n\left(m_{10},l_{10},1,l_{11}\right),y_1^n\right) \in T_\epsilon^{(n)} \textrm{ for } m_{10} \neq 1, \textrm{ and some } l_{10},l_{11}\neq L_{11}\right\},\\
\xi_{44}' &=& \left\{\left(q^n,u_1^n\left(m_{10},l_{10}\right),u_2^n\left(1,l_{20}\right),v_1^n\left(m_{10},l_{10},1,l_{11}\right),y_1^n\right) \in T_\epsilon^{(n)} \textrm{ for } m_{10} \neq 1, \textrm{ and some } l_{10},l_{20}\neq L_{20},l_{11}\neq L_{11}\right\},\\
\xi_{51}' &=& \left\{\left(q^n,u_1^n\left(m_{10},l_{10}\right),u_2^n\left(1,L_{20}\right),v_1^n\left(m_{10},l_{10},m_{11},l_{11}\right),y_1^n\right) \in T_\epsilon^{(n)} \textrm{ for } m_{10} \neq 1,\ m_{11} \neq 1, \textrm{ and some } l_{10},l_{11}\right\},\\
\xi_{52}' &=& \left\{\left(q^n,u_1^n\left(m_{10},l_{10}\right),u_2^n\left(1,l_{20}\right),v_1^n\left(m_{10},l_{10},m_{11},l_{11}\right),y_1^n\right) \in T_\epsilon^{(n)} \textrm{ for } m_{10} \neq 1,\ m_{11} \neq 1, \textrm{ and some } l_{10},l_{11},l_{20}\neq L_{20}\right\},\\
\xi_{61}' &=& \left\{\left(q^n,u_1^n\left(1,L_{10}\right),u_2^n\left(m_{20},l_{20}\right),v_1^n\left(1,L_{10},m_{11},l_{11}\right),y_1^n\right) \in T_\epsilon^{(n)} \textrm{ for } m_{20} \neq 1,\ m_{11} \neq 1, \textrm{ and some } l_{20},l_{11}\right\},\\
\xi_{62}' &=& \left\{\left(q^n,u_1^n\left(1,l_{10}\right),u_2^n\left(m_{20},l_{20}\right),v_1^n\left(1,l_{10},m_{11},l_{11}\right),y_1^n\right) \in T_\epsilon^{(n)} \textrm{ for } m_{20} \neq 1,\ m_{11} \neq 1, \textrm{ and some } l_{20},l_{11},l_{10}\neq L_{10}\right\},\\
\xi_{71}' &=& \left\{\left(q^n,u_1^n\left(m_{10},l_{10}\right),u_2^n\left(m_{20},l_{20}\right),v_1^n\left(m_{10},l_{10},1,L_{11}\right),y_1^n\right) \in T_\epsilon^{(n)} \textrm{ for } m_{10} \neq 1,\ m_{20} \neq 1, \textrm{ and some } l_{10},l_{20}\right\},\\
\xi_{72}' &=& \left\{\left(q^n,u_1^n\left(m_{10},l_{10}\right),u_2^n\left(m_{20},l_{20}\right),v_1^n\left(m_{10},l_{10},1,l_{11}\right),y_1^n\right) \in T_\epsilon^{(n)} \textrm{ for } m_{10} \neq 1,\ m_{20} \neq 1, \textrm{ and some } l_{10},l_{20},l_{11}\neq L_{11}\right\},\\
\xi_8' &=& \big\{\left(q^n,u_1^n\left(m_{10},l_{10}\right),u_2^n\left(m_{20},l_{20}\right),v_1^n\left(m_{10},l_{10},m_{11},l_{11}\right),y_1^n\right) \in T_\epsilon^{(n)} \textrm{ for } m_{10} \neq 1,\ m_{20} \neq 1,\ m_{11}\neq 1, \\ && \textrm{ and some } l_{10},l_{20},l_{11}\big\}.
\end{eqnarray*}
}
The probability of $\xi_{31}'$ can be bounded as follows:
\begin{eqnarray*}
P(\xi_{31}') &=& \sum_{m_{11}=2}^{2^{nR_{11}}}\ \sum_{l_{11}=1}^{2^{R_{11}'}} P\left(\{\left(q^n,u_1^n\left(1,L_{10}\right),u_2^n\left(1,L_{20}\right),v_1^n\left(1,L_{10},m_{11},l_{11}\right),y_1^n\right) \in T_\epsilon^{(n)}\}\right)\\
&\leq& 2^{n\left(R_{11}+R_{11}'\right)} \sum_{(q^n,u_1^n,u_2^n,v_1^n,y_1^n)\in T_\epsilon^{(n)}} p(q^n)p(u_1^n|q^n)p(u_2^n|q^n)p(v_1^n|u_1^n,q^n)p(y_1^n|u_1^n,u_2^n,q^n)\\
&\leq& 2^{n\left(R_{11}+R_{11}'\right)} 2^{-n\left(H(Q,U_1,V_1)+H(U_2|Q)+H(Y_1|U_1,U_2,Q)-H(Q,U_1,U_2,V_1,Y_1)-\delta_3'(\epsilon)\right)}\\
&\leq& 2^{n\left(R_{11}+R_{11}'\right)} 2^{-n\left(I(U_1,V_1;U_2|Q)+I(V_1;Y_1|U_1,U_2,Q)-\delta_3'(\epsilon)\right)},
\end{eqnarray*}
where $\delta_3'(\epsilon)\to 0$ as $\epsilon \to 0$. Obviously, the probability that $\xi_{31}'$ happens goes to $0$ if
\begin{equation}\label{eq_prob_hatxi31}
R_{11}+R_{11}' \leq I(U_1,V_1;U_2|Q)+I(V_1;Y_1|U_1,U_2,Q).
\end{equation}
Similarly, the error probability corresponding to the left error events goes to $0$, respectively, if
\begin{eqnarray}
\label{eq_prob_hatxi32} R_{11}+R_{11}'+R_{20}' &\leq& I(U_1,V_1;U_2|Q)+I(V_1,U_2;Y_1|U_1,Q),\\
\label{eq_prob_hatxi33} R_{11}+R_{10}'+R_{11}' &\leq& I(U_1,V_1;U_2|Q)+I(U_1,V_1;Y_1|U_2,Q),\\
\label{eq_prob_hatxi34} R_{11}+R_{10}'+R_{11}'+R_{20}' &\leq& I(U_1,V_1;U_2|Q)+I(U_1,V_1,U_2;Y_1|Q),\\
\label{eq_prob_hatxi41} R_{10}+R_{10}' &\leq& I(U_1,V_1;U_2|Q)+I(U_1,V_1;Y_1|U_2,Q),\\
\label{eq_prob_hatxi42} R_{10}+R_{10}'+R_{20}' &\leq& I(U_1,V_1;U_2|Q)+I(U_1,V_1,U_2;Y_1|Q),\\
\label{eq_prob_hatxi43} R_{10}+R_{10}'+R_{11}' &\leq& I(U_1,V_1;U_2|Q)+I(U_1,V_1;Y_1|U_2,Q),\\
\label{eq_prob_hatxi44} R_{10}+R_{10}'+R_{11}'+R_{20}' &\leq& I(U_1,V_1;U_2|Q)+I(U_1,V_1,U_2;Y_1|Q),\\
\label{eq_prob_hatxi51} R_{10}+R_{11}+R_{10}'+R_{11}' &\leq& I(U_1,V_1;U_2|Q)+I(U_1,V_1;Y_1|U_2,Q), \\
\label{eq_prob_hatxi52} R_{10}+R_{11}+R_{10}'+R_{11}'+R_{20}' &\leq& I(U_1,V_1;U_2|Q)+I(U_1,V_1,U_2;Y_1|Q), \\
\label{eq_prob_hatxi61} R_{11}+R_{20}+R_{11}'+R_{20}' &\leq& I(U_1,V_1;U_2|Q)+I(V_1,U_2;Y_1|U_1,Q), \\
\label{eq_prob_hatxi62} R_{11}+R_{20}+R_{10}'+R_{11}'+R_{20}' &\leq& I(U_1,V_1;U_2|Q)+I(U_1,V_1,U_2;Y_1|Q), \\
\label{eq_prob_hatxi71} R_{10}+R_{20}+R_{10}'+R_{20}' &\leq& I(U_1,V_1;U_2|Q)+I(U_1,V_1,U_2;Y_1|Q), \\
\label{eq_prob_hatxi72} R_{10}+R_{20}+R_{10}'+R_{11}'+R_{20}' &\leq& I(U_1,V_1;U_2|Q)+I(U_1,V_1,U_2;Y_1|Q), \\
\label{eq_prob_hatxi8} R_{10}+R_{11}+R_{20}+R_{10}'+R_{11}'+R_{20}' &\leq& I(U_1,V_1;U_2|Q)+I(U_1,V_1,U_2;Y_1|Q).
\end{eqnarray}
Note that there are some redundant inequalities in \eqref{eq_prob_hatxi31}-\eqref{eq_prob_hatxi8}: \eqref{eq_prob_hatxi32} is implied by \eqref{eq_prob_hatxi61}; \eqref{eq_prob_hatxi33} is implied by \eqref{eq_prob_hatxi51}; \eqref{eq_prob_hatxi41} is implied by \eqref{eq_prob_hatxi43}; \eqref{eq_prob_hatxi42} is implied by \eqref{eq_prob_hatxi71}; \eqref{eq_prob_hatxi43} is implied by \eqref{eq_prob_hatxi51}; \eqref{eq_prob_hatxi34}, \eqref{eq_prob_hatxi44}, \eqref{eq_prob_hatxi52}, \eqref{eq_prob_hatxi62}, \eqref{eq_prob_hatxi71}, and \eqref{eq_prob_hatxi72} are implied by \eqref{eq_prob_hatxi8}. By combining with the error analysis at the encoder, we can recast the rate constraints \eqref{eq_prob_hatxi31}-\eqref{eq_prob_hatxi8} as:
\begin{eqnarray*}
R_{11} &\leq& I(U_1,V_1;U_2|Q)+I(V_1;Y_1|U_1,U_2,Q)-I(V_1;S|U_1,Q),\\
R_{10}+R_{11} &\leq& I(U_1,V_1;U_2|Q)+I(U_1,V_1;Y_1|U_2,Q)-I(U_1,V_1;S|Q),\\
R_{11}+R_{20} &\leq& I(U_1,V_1;U_2|Q)+I(V_1,U_2;Y_1|U_1,Q)-I(V_1;S|U_1,Q)-I(U_2;S|Q),\\
R_{10}+R_{11}+R_{20} &\leq& I(U_1,V_1;U_2|Q)+I(U_1,V_1,U_2;Y_1|Q)-I(U_1,V_1;S|Q)-I(U_2;S|Q).
\end{eqnarray*}

The error analysis for transmitter $2$ and decoder $2$ is similar to user $1$ and is omitted here. Correspondingly, \eqref{eq_rate_constraint_2_21} to \eqref{eq_rate_constraint_2_24} show the rate constraints for user $2$. Furthermore, the right hand sides of the inequalities \eqref{eq_rate_constraint_2_11} to \eqref{eq_rate_constraint_2_24} are guaranteed to be non-negative when choosing the probability distribution. As long as \eqref{eq_rate_constraint_2_11} to \eqref{eq_rate_constraint_2_24} are satisfied, the probability of error can be bounded by the sum of the error probability at the encoders and the decoders, which goes to $0$ as $n\to\infty$.
\end{proof}
\begin{Remark}
The achievable regions in the above theorems are being further studied in several special cases by only deploying Gel'fand-Pinsker coding for the public message or only for the private message at the transmitters. In addition, the application of special coding schemes to the strong (or weak) state-dependent IC is also under investigation.
\end{Remark}
\begin{Remark}
It can be easily seen that the achievable rate region $\mathcal{R}_1$ in Theorem \ref{theorem_1} is a subset of $\mathcal{R}_2$, i.e., $\mathcal{R}_1 \subseteq \mathcal{R}_2$. However, whether these two regions are equivalent is still under investigation.
\end{Remark}
%
\section{Conclusion}\label{sec_5}
We considered the interference channel with state information non-causally known at both transmitters. Two achievable rate regions are established based on two coding schemes with simultaneous encoding and superposition encoding, respectively. 

\newpage



\end{document}